%% file: main.tex
\documentclass[a4paper]{article}
\usepackage{microtype,fullpage}
\usepackage{amsfonts, amsmath, amssymb, amsthm}
\usepackage{ graphicx, mathtools }
\usepackage{ boxedminipage, etoolbox, framed, ifthen, tabularx,
  tcolorbox, tikz, xifthen }

\newlength{\RoundedBoxWidth}
\newsavebox{\GrayRoundedBox}
\newenvironment{GrayBox}[1]
   {\setlength{\RoundedBoxWidth}{.93\textwidth}
    \def\boxheading{#1}
    \begin{lrbox}{\GrayRoundedBox}
       \begin{minipage}{\RoundedBoxWidth}}
   {   \end{minipage}
    \end{lrbox}
    \begin{center}
    \begin{tikzpicture}%
       \node(Text)[draw=black!20,fill=white,rounded corners,%
             inner sep=2ex,text width=\RoundedBoxWidth]%
             {\usebox{\GrayRoundedBox}};
        \coordinate(x) at (current bounding box.north west);
        \node [draw=white,rectangle,inner sep=3pt,anchor=north west,fill=white]
        at ($(x)+(6pt,.75em)$) {\boxheading};
    \end{tikzpicture}
    \end{center}}
\newenvironment{defproblemx}[2][]{\noindent\ignorespaces%
                                \FrameSep=6pt%
                                \parindent=0pt%
                \vspace*{-1.5em}
                \ifthenelse{\isempty{#1}}{%
                  \begin{GrayBox}{\textsc{#2}}%
                }{%
                  \begin{GrayBox}{\textsc{#2}  parameterized by~{#1}}%
                }
                \begin{tabular*}{\textwidth}{@{\hspace{.1em}} >{\itshape} p{1.8cm} p{0.8\textwidth} @{}}%
            }{
                \end{tabular*}%
                \end{GrayBox}%
                \ignorespacesafterend
            }
\newcommand{\defproblem}[3]{%
  \begin{defproblemx}{#1}
    Input:  & #2 \\
    Task: & #3
  \end{defproblemx}
}%

\usepackage{graphicx}
\usepackage[hidelinks]{hyperref}
\usepackage{mathtools}
\usepackage{tikz}
\usetikzlibrary{calc}
\usepackage{xspace}
\usepackage{booktabs}
\usepackage{multirow}
\usepackage{thm-restate}
\usepackage{cleveref}

\newtheorem{theorem}{Theorem}
\numberwithin{theorem}{section}
\newtheorem{lemma}{Lemma}
\numberwithin{lemma}{section}
\newtheorem{claim}{Claim}
\numberwithin{claim}{section}
\newtheorem{observation}{Observation}
\numberwithin{observation}{section}

\newenvironment{claimproof}{\proof[Proof of claim]}{\endproof}

\theoremstyle{definition}
\newtheorem{definition}{Definition}
\numberwithin{definition}{section}

\theoremstyle{remark}
\newtheorem{remark}{Remark}

\newcommand{\cutenum}[1]{#1_1, \dots, #1_k}
\newcommand{\cutproduct}[1]{#1_1\cup \cdots \cup #1_k}

\newcommand{\Pcal}{\mathcal{P}}
\newcommand{\A}{\mathcal{A}}

\newcommand{\cO}{\mathcal{O}}
\newcommand{\ract}{\mathcal{R}_{\mathrm{act}}}

\newcommand{\minkdmclong}{\textsc{Diverse Min $s$-$t$-Cuts}\xspace}

\newcommand{\minkdmc}{\textsc{Diverse Min $s$-$t$-Cuts}\xspace}

\newcommand{\NP}{\textup{\textsf{NP}}}
\newcommand{\FPT}{\textup{\textsf{FPT}}\xspace}
\newcommand{\POLY}{\textup{\textsf{P}}}

\newcommand{\sharpp}{{\#\textsf{P}}}

\DeclareMathOperator{\cut}{{\sf cut}}
\newcommand{\Mincut}{Minimum $s$--$t$ cut\xspace}
\newcommand{\mincut}{minimum $s$--$t$ cut\xspace}
\newcommand{\mincuts}{minimum $s$--$t$ cuts\xspace}
\newcommand{\maxflow}{maximum $s$--$t$ flow\xspace}

\newcommand{\flownet}{\ensuremath{\mathbb{G}}}

\title{An FPT Algorithm for Diverse Minimum $s$--$t$ Cuts} %

\author{Krishnan Dehaleesan\thanks{University of Bergen, Norway. Email:
\texttt{krishnan.dehaleesan@uib.no}.}
\and
Pål Grønås Drange\thanks{University of Bergen, Norway. Email: \texttt{Pal.Drange@uib.no}.}
\and
Fedor V. Fomin\thanks{University of Bergen, Norway. Email: \texttt{fedor.fomin@uib.no}. Research supported by the Research Council of Norway under BWCA project (grant no.~314528) and by the European Research Council (ERC) under the European Union's Horizon 2020 research and innovation programme (NewPC grant agreement No.  101199930)}
\and
Petr A. Golovach\thanks{University of Bergen, Norway. Email: \texttt{petr.golovach@uib.no}. Research supported by the Research Council of Norway under the BWCA (grant no.~314528) and  Extreme-Algorithms (grant no~355137) projects.}
\and
Laure Morelle \thanks{University of Bergen, Norway. Email: \texttt{Laure.Morelle@uib.no}. Research supported by the Research Council of Norway under the  Extreme-Algorithms (grant no~355137) project.}}

\date{}

\begin{document}

\maketitle

\begin{abstract}
We study the problem of finding a family of diverse minimum edge $s$--$t$ cuts  in a directed weighted graph $G$.
Given integers $k$ and $d$, the task is to decide whether $G$ contains $k$ minimum $s$--$t$ cuts
$C_1, \dots, C_k$ such that for any $i,j \in [k]$, the number of edges in the symmetric difference
$C_i \triangle C_j$ is at least $d$.

For $d \in\{1,2\}$, the problem corresponds to counting minimum $s$--$t$ cuts  in $G$, which is
$\sharpp$-complete [Provan and Ball, SICOMP 1983].
The problem is also known to be $\NP$-complete already for $k = 3$
[de Berg, López Martínez, Spieksma, ISAAC 2024].
Our main result shows that the problem is fixed-parameter tractable (FPT) when parameterized by the combined parameter $k + d$.

The main ingredients of our FPT algorithm build on novel structural properties of diverse minimum $s$--$t$ cuts
and a non-trivial application of the flow-augmentation technique of
Kim, Kratsch, Pilipczuk, and Wahlström [JACM 2025].

\end{abstract}

\section{Introduction}

The {\sc \mincut} problem is one of the most fundamental and central problems in graph algorithms.
Given a directed graph with designated vertices \(s\) and \(t\), the task is to separate \(s\) from \(t\) by removing a set of edges of minimum total capacity.
This problem is deeply connected to the theory of network flows, as formalized by the max-flow min-cut theorem, and to notions of graph connectivity.

It is well known that the {\sc \mincut} problem is solvable in polynomial time via the connection to the \textsc{Maximum Flow} problem, and a long line of research has culminated in an almost-linear time $O(m^{1+ o(1)})$~\cite{chen2025maximumflow}.
However, several natural variants quickly become computationally more challenging.
In particular, counting the number of \mincuts is \sharpp-hard \cite{provan1983complexitycounting}. On the other hand,  the maximum number of
\emph{disjoint} \mincuts could be easily found in polynomial time \cite{wagner1990disjointcuts}. (One proceeds iteratively by computing the leftmost minimum cut and contracting its edges.)

In applications, one often wants not just a single minimum $s$--$t$ cut, but several optimal alternatives from which a user can choose. A naive approach may return cuts that are almost identical, differing in only a few edges. It is therefore natural to ask for a family of minimum $s$--$t$ cuts that is diverse, meaning that every pair differs in at least $d>0$ edges.
This perspective fits into the broader study of diversity in combinatorial optimization
\cite{baste2022diversitysolutions,fomin2024diversepairs,fomin2024diversecollections,hanaka2021findingdiverse,deberg2026disjointtours}. In this context, de~Berg et al.~\cite{deberg2024findingdiverse} introduced the notion of \emph{diverse} \mincuts. The problem of identifying diverse \mincuts occupies an ``intermediate'' position between counting all minimum cuts (\sharpp-complete) and computing pairwise disjoint minimum cuts (in \POLY).

For a  family of  \mincuts  $C_1, C_2, ..., C_k\subseteq E(G)$ in a graph $G$, de~Berg et al. define the following natural notions of their diversity:

\medskip
\begin{tabular}{@{}llp{0.55\linewidth}@{}}
Name & Objective & Interpretation \\
\midrule
min & $\min_{i,j}\lvert C_i \triangle C_j\rvert$ &
Maximize the minimum pairwise ``disjointness''. \\
cov & $\left\lvert \bigcup_i C_i \right\rvert$ &
Maximize coverage of the solution space. \\
sum & $\sum_{i,j}\lvert C_i \triangle C_j\rvert$ &
Maximize the total pairwise ``disjointness''. \\
\end{tabular}

\medskip
\noindent
Here $C_i \triangle C_j$ denotes the symmetric difference of the edge sets of cuts $C_i $ and  $C_j$.

By making use of submodular function minimization, de~Berg et al.~\cite{deberg2024findingdiverse} provided a (strongly) polynomial-time algorithm for computing \mincuts with respect to the \emph{cov} and \emph{sum} measures of diversity.
Interestingly, the computational complexity of the \emph{min} diversity measure appears to be more elusive.
As shown by de~Berg et al.~\cite{deberg2024findingdiverse}, the problem of finding $k=3$ \mincuts under the \emph{min} diversity measure is already \(\NP\)-hard.
Meanwhile, when $k=2$, the three measures converge, and the solution is always to take the two extremal cuts, i.e., the minimum $s$--$t$ cut closest to $s$, and the minimum $s$--$t$ cut closest to $t$ \cite{deberg2024findingdiverse}.
Our main results shed light on the parameterized complexity of this intriguing measure of diversity.

\medskip

\subsection{Overview of Our Results}

We consider a more general capacitated (weighted) variant of the diverse \mincuts problem defined by de~Berg et al.~\cite{deberg2024findingdiverse}.
While the weight of a cut is measured by the total capacity of its edges, the diversity between cuts is not measured by capacities, but rather by the number of edges in the symmetric difference of their edge sets.
Because of that, it is convenient to formulate the problem using two cost functions on the edges of the input graph.
The first cost function \(c\), which we call the \emph{capacity}, represents the weight of an edge and contributes to the total weight of a cut.
The second cost function, called the \emph{multiplicity}, arises from the algorithmic process.

In our algorithm, we contract certain edges of the graph, which may result in a graph containing parallel edges.
Since we subsequently delete parallel edges, it is necessary to preserve their potential contribution both to cut weights and to the diversity measure.
To maintain the correct contribution of a set of parallel edges to a cut, we update the capacity of the resulting edge accordingly.
However, since the diversity of cuts is defined via the symmetric difference of their edge sets, we must additionally maintain the multiplicity $\ell$  of edges.

Here is the formal definition of the problem.

\defproblem%
% NAME
{\minkdmclong{}}%
% INPUT
{A flow network
  $\flownet = (G=(V,E), c: E(G)\to\mathbb{Z}_{\geq 1}, \ell: E(G) \to
  \mathbb{Z}_{\geq 1}, s, t)$, where $c$ is the \emph{capacities} and
  $\ell$ the \emph{multiplicities} of edges, and $s$ and $t$ are two
  terminals, together with positive integers $d$ and $k$.}%
% QUESTION
{Decide whether there exists a family of $k$ minimum capacity $s$--$t$
  cuts $C_1, \dots, C_k \subseteq E$ such that, for each
  $i,j \in [k], \ \ell(C_i \bigtriangleup C_j) \ge d$.}

\noindent
With unit capacities and multiplicities, this corresponds to the \emph{min} measure studied by de~Berg et al.~\cite{deberg2024findingdiverse}.

For $d\in\{1,2\}$, this is the problem of finding $k$ distinct minimum cuts, but
when $d$ is part of the input, the problem is $\NP$-complete, even for $k=3$~\cite{deberg2024findingdiverse}.
For unweighted graphs (graphs with unit capacities and multiplicities), by putting~$d$ to be twice the size of the minimum cut, we obtain  the problem of  computing of $k$ pairwise disjoint minimum cuts, which is solvable in polynomial time~\cite{deberg2024findingdiverse,wagner1990disjointcuts}.

\subsubsection{Outline of the Algorithm}

In our main result, we show that the problem is \FPT when parameterized by $k$
and $d$.

\begin{restatable}{theorem}{theoremdiversecuts}
\label{thm:fpttime}
\minkdmclong can be solved in time
$(kd)^{\cO(k^8d^8)}\cdot n^{\cO(1)}$.
\end{restatable}

\noindent
The algorithm for \Cref{thm:fpttime} is constructive. If the input graph contains $k$ $d$-diverse \mincuts, the algorithm finds these cuts within the same running time.

The algorithm is obtained by combining novel structural properties of diverse \mincuts
and the flow-augmentation technique of Kim et al.~\cite{kim2025flow}.
In the first step of the algorithm, we preprocess the input in polynomial time and either solve the problem or obtain an equivalent instance, where the number of edges in all  \mincuts  is  bounded.
The second step of the algorithm  solves the problem when the number of edges in the minimum cut is bounded.

\begin{figure}[t]
\begin{center}
\scalebox{0.8}{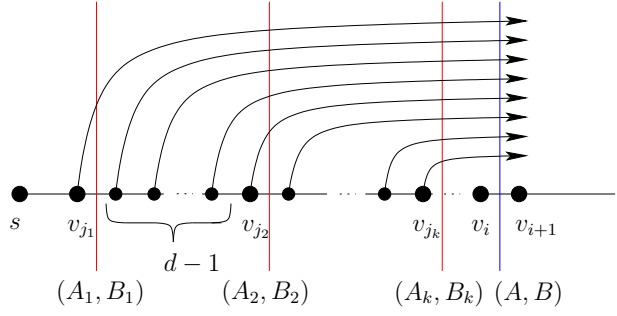}
\caption{Construction of $k$ $d$-diverse cuts; the cuts $(A,B)$ and $(A_1,B_1),\ldots(A_k,B_k)$ are shown by vertical lines. Note that it may happen that $s=v_{j_1}$ or $v_{j_k}=v_i$.}\label{fig:kdsquaredyes_2}
\end{center}
\end{figure}

\Cref{lem:preprocess,lem:mainstep} together provide the proof of \Cref{thm:fpttime}. Indeed, \Cref{lem:preprocess} promises to find $k$ $d$-diverse cuts or output an equivalent instance in which all the minimum cuts have a size bounded by $(kd)^2$. In the latter case, we make use of \Cref{lem:mainstep} that guarantees an \FPT algorithm when the size of the \mincuts are bounded. Substituting the bound $r=(kd)^2$, \Cref{thm:fpttime} follows.

\medskip We briefly sketch the main ideas of the proof  of
\Cref{lem:preprocess} for the case when initially, both the capacity and multiplicity functions have unit values. The actual proof for the general case hinges on the same ideas but is more technical.

Given a directed graph $G$ with a source and sink vertices $s$ and $t$, respectively, we identify the \emph{critical} edges, that is, edges that are included in some \mincut. Then we preprocess the instance by contracting some \emph{residual}, i.e. non-critical, edges. By contractions, we obtain a new instance where the capacities and multiplicities get modified. If an edge-contraction creates multiple edges, we replace them by a single edge whose capacity and multiplicity is the sum of these values of the multiple edges.

We want to contract edges in such a way that \mincuts are not changed. While for
  undirected graphs, contracting all residual edges does not change \mincuts, this is not true for directed graphs. Here we have to be more careful. We use classical Menger's theorem (see, e.g, \cite{diestel_graph_2016,BJensenG09}) that $G$ has an $s$--$t$ cut of size $\mu$ if and only if $G$ has $\mu$ edge-disjoint $s$--$t$ paths $P_1,\ldots,P_\mu$. First, we contract residual edges of these paths. Then we contract the edges of all strongly connected components of $G-\bigcup_{i=1}^\mu E(P_i)$. Finding diverse \mincuts in the  reduced instance is equivalent to finding  diverse \mincuts  in the original one.

The reduced graph \(G^*\) has a very particular structure.
If we reverse the directions of all residual edges of \(G^*\), we obtain a directed acyclic graph   \(\tilde{G}\).
An interesting, and somewhat surprising, property of \(\tilde{G}\) is that for every minimum cut \(C=E(A,B)\) of \(G^*\), there exists a topological ordering \(v_1, \ldots, v_n\) of \(\tilde{G}\) such that
\(A = \{v_1, \ldots, v_i\}\) for some \(i \in [n]\).
Moreover, for each \(i \in [n]\), the partition \((A,B)\) with %defined by
\(A = \{v_1, \ldots, v_i\}\) and \(B = \{v_{i+1}, \ldots, v_n\}\)
defines
%is itself
a minimum cut in \(G^*\).

 By leveraging these properties of cuts in \(G^*\), we can conclude that if the preprocessed instance admits a minimum cut with at least \((kd)^2\) edges, then the graph contains \(k\) minimum \(s\)--\(t\) cuts that are pairwise \(d\)-diverse.
Indeed, since a minimum cut \(C=E(A,B)\) with \((kd)^2\) edges can be expressed as
\(A = \{v_1, \ldots, v_i\}\) and \(B = \{v_{i+1}, \ldots, v_n\}\)
for some topological ordering of \(\tilde{G}\), at least \(d(k-1)+1\)
endpoints of the cut edges must lie in either \(A\) or \(B\). Indeed, otherwise with $d(k-1)$ points on either side, the maximum number of edges from A to B is limited to $d^2(k-1)^2 \le (kd)^2$.

Without loss of generality, suppose that \(A\) contains at least \(d(k-1)+1\) such endpoints.
Since for every \(j \leq i\), the partition
\(A' = \{v_1, \ldots, v_j\}\) and \(B' = \{v_{j+1}, \ldots, v_n\}\)
also defines a minimum cut, we can select \(k\) minimum cuts whose pairwise symmetric differences are of size at least \(d\).

 The proof of \Cref{lem:preprocess} is constructive and allows us to find the actual cuts. Because a \mincut with maximum number of edges can be found in polynomial time, the problem is either solved or is reduced to the case when every \mincut contains less than $(kd)^2$ edges.

\medskip
The algorithm in the proof of \Cref{lem:mainstep} consists of three major steps.
First, we use the \emph{color coding} technique of Alon, Yuster, and Zwick~\cite{alon_color-coding_1995} to highlight edges in different cuts. The aim here is to find each of the cuts separately.
For this, notice that the edges of $\bigcup_{i=1}^kC_i$ for a hypothetical solution $C_1,\ldots,C_k$ can be partitioned into $2^k-1$ classes according to their inclusion to cuts, and at most $rk$ classes are nonempty because each cut contains at most $r$ edges.
Hence, we can guess nonempty classes and highlight their edges.
By this step, the problem boils down to finding a \mincut with specific constraints on the colors of its edges and their multiplicities.

In the second step, we find such cuts using the flow augmentation technique. The main result of Kim et al.~\cite{kim2025flow} (see~\Cref{thm:flowaugdeterm}) is that it is possible to augment a directed graph by adding some edges in such a way that an arbitrary minimal $s$--$t$ cut becomes minimum.  Notice that we are looking for cuts of size at most $r$ in the (non-capacitated) graph $G$ that are inclusion-minimal. Thus, by making use the flow-augmentation technique, we further reduce the problem to the case when we have to find a certain \mincut which remains minimum if we forget about the edge capacities.

Finally, to solve the obtained problem, we again use Menger's theorem. We find a maximum family of at most $r$ edge-disjoint $s$--$t$ paths in $G$ and make guesses about the colors and multiplicities of the edges from a solution cut in each path. After that, we can check whether there is a \mincut satisfying the given constraints using the standard \mincut algorithm.

\subsection{Related work}

The structure of minimum cuts has been studied extensively, beginning with the classical work of Picard and Queyranne, who characterized the full family of minimum cuts in a network and established their lattice structure~\cite{picard1982structureall}. From an algorithmic perspective, Karger and Stein gave a randomized algorithm for computing global minimum cuts~\cite{karger_new_1996}. Wagner studied a related but distinct notion of diversity by considering the problem of finding $k$ pairwise-disjoint $s$--$t$-cuts of minimum total cost in edge-weighted graphs~\cite{wagner1990disjointcuts}. More broadly, diversity has appeared in optimization through the maximum diversity problem, which was analyzed and modeled via zero--one programming by Kuo et al.~\cite{kuo1993analyzingmodeling}.

Recent work has explored diversity across many combinatorial problems, including minimum cuts~\cite{deberg2024findingdiverse}, matchings~\cite{fomin2024diversepairs}, vertex cover~\cite{baste2022diversitysolutions}, shortest paths~\cite{hanaka2022computingdiverse}, spanning trees~\cite{hanaka2021findingdiverse}, and travelling salesman~\cite{deberg2026disjointtours}. Most relevant to our setting, Hanaka et al.\ recently showed that finding $k$ pairwise-disjoint global minimum cuts is \NP-hard~\cite{hanaka2023frameworkdesign}, sharply contrasting with the tractability of classical minimum cut problems and motivating the study of alternative notions of diversity among minimum cuts.

Specifically, as mentioned above, de Berg et al.~\cite{deberg2024findingdiverse} consider all three diversity measures for \minkdmclong, proving that the problem is solvable in polynomial-time with respect to the \emph{cov} and the \emph{sum} measure, but \NP-hard with respect to the \emph{min} measure. They also provide a polynomial-time algorithm for computing disjoint \mincuts.

Drabik and Masařík show that diverse solutions for a wide class of vertex problems can be computed in \FPT{} time parameterized by cliquewidth by extending dynamic programs on cliquewidth decompositions to support general diversity measures~\cite{drabik2026findingdiverse}.

\section{Preliminaries}

We consider directed graphs with neither loops nor multiple edges. The vertex set of a graph $G$ is denoted by $V (G)$ and the edge set by $E(G)$. We set $n = |V (G)|$ and $m = |E(G)|$.
We refer the reader to~\cite{ahuja1993networkflows} for the definition of \emph{flow} and related terminology.
We assume the reader is familiar with basic parameterized complexity, but for terminology and concepts on parameterized algorithms and complexity not defined in this article, we refer the reader to the textbook by Cygan et al.~\cite{cygan2015parameterized}.
Throughout the paper, $c : E \to \mathbb{Z}_{\geq 1}$ is the edge \emph{capacity} and $\ell : E \to \mathbb{Z}_{\geq 1}$ is the edge \emph{multiplicity} (which can alternatively be seen as the number of parallel edges). For a set $X \subseteq E(G)$, we extend $\ell$ additively by defining $\ell(X) = \sum_{e \in X}\ell(e)$.
A \emph{flow network} is a tuple $\flownet = (G, c, \ell, s, t)$, where $s$ is the source and $t$ the target.
We will refer directly to $G = (V, E)$ in the flow network when it is clear from context.
For $a<b \in \mathbb{Z}$, we define $[a,b]= \{a,a+1,\dots, b\}$ and $[b]=[1,b]$.

\begin{definition}[\Mincut]
Let $\flownet$ be a  flow network and $d,k\in \mathbb{Z}_{\geq 1}$.
  An
  \emph{$s$-$t$ cut}, denoted $(A,B)$, is a partitioning of the
  vertices $V$ into $A, B= V \setminus A$ such that $s \in A$ and
  $t \in B$.  We will denote by $E(A, B)$ the set of edges going from
  $A$ to $B$. Slightly abusing notation, we also refer to the edge set $C = E(A, B)$
  as the cut.  The \emph{capacity} of a cut is
  \[
    c(A,B)=c(E(A,B)) = \sum_{e \in E(A,B)} c(e).
  \]
  The cut $(A,B)$ is a \emph{\mincut} if it is a cut of minimum capacity, whose value we denote by $\mu(G)$.
  Equivalently, a $s$-$t$-cut is minimum if its capacity is equal
  to the maximum flow value.
\end{definition}

\begin{definition}[$d$-diverse cut]
Two cuts $C,C'$ are said to be \emph{$d$-diverse} if $\ell(C \triangle C')\ge d$, where $\triangle$ denotes symmetric difference.
The cuts $C_1,\dots,C_k$ are said to be \emph{$d$-diverse} if they are pairwise $d$-diverse.
\end{definition}

\subparagraph{Conventions.}
Note that every edge has a capacity and multiplicity of at least one.  We will
assume that the network has a non-zero maximum flow, i.e., there is a
path from $s$ to $t$.  We also assume that $s$ is a source-vertex and
that $t$ is a sink-vertex, i.e., $\deg^-(s) = 0$ and $\deg^+(t) = 0$.
When we say minimum cuts, we refer to \mincuts unless stated otherwise.
We start by recalling some preliminary properties of \mincuts.

\begin{lemma}[Flow-value Lemma~\cite{kleinberg2006algorithm}]
  \label{lem:flowvaluelem}
  For any $s$--$t$-cut $(A, B)$, and flow $f$, the value of the total flow
  $$v(f) = \sum_{e \in E(A,B)}f(e) -\sum_{e \in E(B,A)}f(e).$$
\end{lemma}

\begin{restatable}{lemma}{lemmamengerdecomposition}
  \label{lem:mengerdecompostion}
  Let $\mu(G)$ be the size of a minimum $s$--$t$ edge cut in a graph $G$ with
  unit capacities. Then, it is possible to compute $\mu(G)$ edge-disjoint
  $s$--$t$ paths in $\cO(\min\{m^{1/2},n^{2/3}\}m)$ time.
\end{restatable}

\begin{proof}
This follows directly from the edge version of Menger's theorem
(see, e.g.,~\cite{BJensenG09,diestel_graph_2016}). Alternatively, using Dinic's algorithm~\cite{ahuja1993networkflows},
the maximum $s$--$t$ flow in a directed unit-capacity network can be computed deterministically in
$\cO(\min\{m^{1/2},n^{2/3}\}m)$ time.

Extracting the paths: the max flow is integral (0 or 1 on each edge). Repeatedly start at $s$, follow edges with flow 1 until you reach $t$; record that path and remove the edges in it and repeat until no such path remains. This takes $\cO(m)$ total over all paths.
\end{proof}

\section{Preprocessing}
\label{sec:preprocess}

In this section, we prove the following important lemma, that we can preprocess
our input graph in such a way that it has at most $(kd)^2$ edges.

\begin{lemma}
  \label{lem:preprocess}
  There is an algorithm that, given an instance $(\flownet,d,k)$ of
  \minkdmc, outputs in polynomial time either:
  \begin{itemize}
  \item that $(\flownet,d,k)$ is a yes-instance of \minkdmc, or
  \item an equivalent instance $(\flownet^*,d,k)$ of \minkdmc such that
    any \mincut in $G^*$ has at most $(kd)^2$ edges.
  \end{itemize}
\end{lemma}

\subsection{Preliminary Results}

Let us first prove some small preliminary results. First, we recall the folklore result that the \mincut with maximum number of edges can be found in polynomial time; we provide the proof for completeness.

\begin{lemma}\label{prop:findingmaxedgemincut}
    There is an algorithm that, given a directed graph $G$, a capacity $c:E(G)\to\mathbb{Z}_{\geq 1}$, and terminals $s,t\in V(G)$, outputs a \mincut with the maximum number of edges in polynomial time.
\end{lemma}
\begin{proof}
    Let $M = m+1$.
    Let $c':E(G)\to\mathbb{Z}_{\geq 1}$ be the capacity such that $c'(e) = c(e)\cdot M - 1$ for each edge $e$.
    For a cut $(A,B)$, its capacity with respect to $c'$ is \[\sum_{e \in E(A,B)} c'(e) = M\cdot\sum_{e \in E(A,B)}c(e) -  |E(A,B)|.\]
    Thus, a \mincut of $G$ with respect to $c'$ is a \mincut of $G$ with respect to $c$ that maximize its number of edges.
    Thus, it suffices to compute \mincut of $G$ with respect to $c'$, which can be done in polynomial time~\cite{chen2025maximumflow}.

\end{proof}

We observe that edge contractions preserve flows. Let $G$ be a graph with given source and sink vertices $s$ and $t$, respectively, and a capacity function $c\colon E(G)\rightarrow\mathbb{Z}_{\geq 1}$. Let also $f$ be a feasible $s$--$t$ flow. Consider an arbitrary edge $(u,v)\neq(s,t)$. As $f\colon E(G)\rightarrow \mathbb{Z}_{\geq 0}$, we can use the same convention as with other weight functions, and  assume that if the contraction of $(u,v)$ creates multiple edges, then we define the flow on the resulting single edge by taking the sum of the flows on the multiple edges. We say that the obtained function is an \emph{induced flow}.

\begin{lemma}
\label{claim:inducedflowcontract}
 Let $G=(V,E)$ be a directed $s,t$ flow network with capacities, and let $f$ be a feasible $s$--$t$ flow in $G$. Let $e=(u,v)\in E$ be an edge distinct from $(s,t)$.
 Let $G'$ be the graph obtained from $G$ by contracting $e$, and let $f'$ the induced flow. Then $f'$ is a valid $s$--$t$ flow and $v(f') = v(f)$.
 \end{lemma}

 \begin{proof}
 Notice that the flow conservation at the  single vertex obtained from $u$ and $v$ is preserved. Thus, $f'$ is a flow. As $(u,v)\neq(s,t)$,  $v(f') = v(f)$.

\end{proof}

Let $G$ be a directed graph with edge capacity $c$ and let $s,t\in V(G)$ be terminals.
Let $C$ be the set of all edges that appear in some \mincut, and $R = E\setminus C$. We call $C$ to be the set of \emph{critical edges} of the graph and $R$ to be the set of \emph{residual edges} of the graph.
Let $G_R = (V(G),R)$ represent the residual graph.

By the well-known properties of \mincuts and maximum flows, we can make the following observation.

\begin{observation}\label{obs:saturatedcritical}
  If an edge is critical, then it is saturated in every maximum flow. Furthermore, an $e$ edge is critical if and only if its deletion reduces the maximum flow by the capacity of $e$.
\end{observation}

Recall also the following property of \mincuts and maximum flows.

\begin{observation}\label{obs:backedgenoflow}
    Let $(A,B)$ be a \mincut in a graph $G$ and let $f$ be a maximum $s$--$t$ flow.
    Then every edge $(u,v)\in E(B,A)$ carries zero flow.
\end{observation}
\begin{proof}
    Since $f$ is a maximum flow and $(A,B)$ is a \mincut, we have $v(f) = c(A,B)$. It follows from \Cref{obs:saturatedcritical} (since $(u,v)\in E(A,B)$ edges are critical) that
\[
\sum_{u\in A,\,v\in B} f(u,v) = c(A,B)
\quad\text{and}\quad
\sum_{u\in B,\,v\in A} f(u,v) = 0.
\]
This means that every edge from $B$ to $A$ carries zero flow.
\end{proof}

Finally, we observe that strongly connected components of the residual graph are not separated by \mincuts.

\begin{lemma}\label{claim:scconeside}
    Let $(S,T)$ be a \mincut. Let $G_{R'}= (V',R')$ be a strongly connected component of $G_R$. Then $V' \subseteq S$ or $V' \subseteq T$.
\end{lemma}

\begin{proof}
    Assume toward a contradiction that $V'$ has vertices in $S$ and $T$. This means that we have edges $e$ and $f$ that go from $S$ to $T$ and $T$ to $S$ respectively, since $G_{R'}$ is a strongly connected component. However, this means that $e$ belongs to the \mincut induced by $(S,T)$ which contradicts the fact that $e$ is a residual edge.
\end{proof}

\subsection{Reducing Large Cuts}
\label{sec:G*}

Let $(\flownet,d,k)$ be an instance of \minkdmc.
Recall that $C$ is the set of critical edges, $R$ is the set of residual edges, and that $G_R = (V(G),R)$ is the residual graph.
We construct the instance $(\flownet^*,d,k)$ of \minkdmc of \Cref{lem:preprocess} as follows.
We compute an $s$--$t$ max-flow $f$. We contract all the edges in $R$ that have positive flow through it.

After this, we contract all the strongly connected components of $G_R'$, where $G_R'$ is the graph obtained from $G_R$ by the above contractions. If multiple edges are formed, we replace them with one edge whose capacity (resp. multiplicity) is the sum of the capacities (resp. multiplicities) of these multiple edges.
For the vertices obtained from $s$ and $t$ by the contraction, we keep the same names $s$ and $t$, respectively.
Let \hypertarget{f}{$f^*$} be the flow obtained from $f$ induced on $G^*$.
As the residual edges with a positive flow were contracted, we can make the following remark.

\begin{remark}\label{remark:zeroflowr}
    Edges in $R$ have zero flow in the flow $f^*$.
\end{remark}

The fact that $(\flownet,d,k)$ and $(\flownet^*,d,k)$ are equivalent instances of \minkdmc follows from the next two lemmata. Informally, we first show that the \mincuts of $G^*$ are exactly the \mincuts of $G$.

\begin{restatable}{lemma}{mincutsaftercontraction}
  \label{lem:mincutsaftercontraction}
  For every \mincut $(A,B)$ of $G$, no edge with its endpoints in $A$ and $B$
  is contracted in the construction of $G^*$, and $(A^*,B^*)$ is a \mincut in
  $G^*$, where $A^*$ and $B^*$ are the sets obtained from $A$ and $B$,
  respectively, by the contraction. Furthermore, for any \mincut $(A^*,B^*)$ of
  $G^*$, $(A,B)$ is a \mincut in $G$, where $A^*$ and $B^*$ are obtained from
  $A$ and $B$.
\end{restatable}

\begin{proof}
Denote by $G'$ the graph obtained by the contraction of all the edges in $R$ that have positive flow with respect to $f$, and denote by $c'$ the obtained capacity function.
 Consider an arbitrary minimum cut  $(A,B)$ in $G$. Then every edge $(u,v)$ with $u\in A$ and $v\in B$ is critical and, therefore, is not contracted. For every edge $(u,v)$ with $u\in B$ and $v\in A$, $f(u,v)=0$ by \Cref{obs:backedgenoflow}. Thus, these edges are also not contracted. This implies that $(A',B')$, where $A'$ and $B'$ are the sets of vertices obtained from $A$ and $B$, respectively, by the contractions, is a cut in $G'$. Also, we never contract an edge between $s$ and $t$. Then the induced flow $f'$ is feasible  by \Cref{claim:inducedflowcontract} and $v(f')=v(f)$. As $c'(A',B')=c(A,B)=v(f)$, we have that $f'$ is a maximum flow in $G'$ and $(A',B')$ is a \mincut in $G'$

Let  $(A',B')$ be a \mincut of $G'$. Consider the sets of $A$ and $B$ obtained by ``uncontracting''
of the edges contracted in the construction of $G'$. Since $c'(A',B')=v(f')=v(f)=c(A,B)$, we have that $(A,B)$ is a \mincut of $G$.

This proves that the \mincuts of $G'$ are exactly the \mincuts of $G$.

Hence, it suffices to prove that the \mincuts of $G^*$ are exactly the \mincuts of $G'$.

By~\Cref{claim:scconeside}, for any \mincut $(A',B')$ of $G'$, each contraction happens either
in $A'$ or in $B'$. Therefore, $f^*$ is a maximum flow in $G^*$, and
$(A^*,B^*)$ is a \mincut of $G^*$, where
$A^*$ and $B^*$ are the sets obtained from $A$ and $B$, because $c^*(A^*,B^*)=v(f^*)$.
If $(A^*,B^*)$ is a \mincut of~$G^*$, then by the same argument as above,
$(A',B')$ is a \mincut in~$G'$, where~$A^*$ and~$B^*$ the sets  obtained from~$A'$ and~$B'$.
This concludes the proof.

\end{proof}

\Cref{lem:mincutsaftercontraction} establishes the bijection~$\varphi$ between \mincuts in~$G$ and~$G^*$, where $\varphi(A,B)=(A^*,B^*)$. By the definition of the multiplicity function~$\ell^*$, we immediately obtain the following property.

\begin{lemma}\label{lem:equiv}
  Let  $(A_1,B_1)$ and $(A_2,B_2)$
  be two \mincuts of $G$, and let
  $\varphi(A_1,B_1) = (A_1^*,B_1^*)$ and $\varphi(A_2,B_2) = (A_2^*,B_2^*)$.
  Then
  \[
    \ell(C_1\triangle C_2) = \ell(C_1^*\triangle C_2^*),
  \]
  where
  $C_i = E_G(A_i,B_i)$ and
  $C_i^*=E_{G^*}(A_i^*,B_i^*)$ for $i \in \{1, 2\}$.
\end{lemma}

\noindent
This proves the equivalence of  $(\flownet,d,k)$ and $(\flownet^*,d,k)$.

\bigskip

Since $(\flownet,d,k)$ and $(\flownet^*,d,k)$ are equivalent, from now on, we can assume that~$s$ is a source and~$t$ is a sink in~$G^*$. Otherwise, we can simply delete the edges of the form $(v,s)$ and $(t,v)$ as they are irrelevant.
Let~$\tilde{G}$ be the graph obtained from~$G^*$ by reversing the direction of residual edges.
Hence, the edge set of~$\tilde G$ is~$C\cup \tilde R$, where~$\tilde R$ represent the reversed residual edges.
We use~$\tilde{G}$ to prove some nice properties of~$G^*$.

\begin{lemma}\label{lemma:tildeGdag}
    $\tilde{G}$ is a directed acyclic graph (DAG).
\end{lemma}
\begin{proof}
  Toward a contradiction, consider the existence of a directed cycle~$D$.
  Given that all strongly connected components of~$G_R$ have been
  contracted, it implies that~$(V(\tilde{G}),\tilde{R})$ is a DAG.
  Therefore,~$D$ contains at least one edge from~$C$.
  Since~$e$ is critical in ~$G^*$, there exists a \mincut $(A,B)$ in
  $G^*$ containing~$e$.  By \Cref{obs:backedgenoflow},~$e$ is an edge
  from~$A$ to~$B$.
  Since $D$ is a cycle, there exists an edge $f\in D$ from $B$ to $A$.
  By \Cref{obs:backedgenoflow}, $f \notin C$ and thus $f\in\tilde{R}$.
  Thus, the edge $\tilde{f}$ corresponding to~$f$ in~$G^*$ is a residual
  edge from~$A$ to~$B$.

  Therefore $\tilde{f}$ belongs to $(A,B)$, contradicting the fact that
  $\tilde{f}$ is a residual edge.
\end{proof}

A \emph{topological ordering} of a directed %
graph $G $ is an ordering of the vertices of $G$, $v_1, \dots, v_n$,
such that for every edge $v_iv_j$, $i<j$.  It is well-known that $G$
admits a topological ordering if and only if $G$ is a
DAG~\cite{BJensenG09}.  Using \Cref{lemma:tildeGdag}, we can prove the
following.

\begin{restatable}{lemma}{lemmaordering}
\label{lem:ordering}
Let $(A,B)$ be a \mincut in $G^*$. There is an ordering $s=v_1,v_2, \dots, v_n=t$ of $V(G^*)$ such that the following hold:
\begin{enumerate}
    \item For each $1 \le i \le n$, $(\{v_1, \dots, v_i\}, \{{v_{i+1}},\dots, v_n\})$ is a \mincut.
    \item There exists $1\le j\le n$ such that $A=\{v_1, \dots, v_j\}$ and $B=\{{v_{j+1}},\dots, v_n\}$.
\end{enumerate}
\end{restatable}

\begin{proof}
       Since $(A,B)$ is a \mincut, all edges in $E(A,B)$ belong to $C$ and all edges in $E(B,A)$ belong to $R$ (from Observations~\ref{obs:saturatedcritical} and~\ref{obs:backedgenoflow}). This means that in the graph $\tilde{G}$, $E(B,A) = \emptyset$. Since $\tilde G$ is a DAG by \Cref{lemma:tildeGdag} and that no edge goes
  left through a \mincut, there exist a topological ordering of $\tilde{G}[A]$
  and a topological ordering of $\tilde{G}[B]$ such that their concatenated ordering is a topological ordering $s=v_1,v_2, \dots, v_n=t$
  of $\tilde G$. Notice that
  we can assume that $s$ is the first and $t$ is the last vertex in the ordering because $s$ is a source and $t$ is a sink.
  Hence, there exists $1\le j\le n$ such that $A=\{v_1, \dots, v_j\}$ and $B=\{{v_{j+1}},\dots, v_n\}$.

Let us now prove the first point.
Let $1\le i\le n$, $A_i=\{v_1, \dots, v_i\}$ and $B_i=\{{v_{i+1}},\dots, v_n\}$.
Clearly, $s \in A_i$ and  $t \in B_i$, so $(A_i,B_i)$ is an $s$-$t$-cut.
By \Cref{lem:flowvaluelem},
$v(f) = \sum_{e \in E(A_i,B_i)} f(e) - \sum_{e \in E(B_i,A_i)} f(e)$.
From the definition of the topological ordering, any edge in $e \in E(A_i,B_i)$ belongs to $C$ and any edge in $e \in E(B_i,A_i)$ belongs to $R$ and thus has zero flow by \Cref{remark:zeroflowr}. It follows that
  \begin{align*}
    v(f) &= \sum_{e \in E(A_i,B_i)} f(e) - \sum_{e \in E(B_i,A_i)} f(e)\\
         &= \sum_{e \in E(A_i,B_i)} f(e) - 0\\
         &= \sum_{e \in E(A_i,B_i)} c(e)
           = c(A_i,B_i),
  \end{align*}

  so $\cut(A_i,B_i)$ is an $s$-$t$-cut whose capacity is equal to the
  maximum flow, and thus it is a \mincut.
\end{proof}

\subsection{Correctness of the Preprocessing}

Using \Cref{lem:ordering}, we can prove that if $G^*$ admits a minimum cut with many edges, then $G^*$ is a yes-instance.

\begin{restatable}{lemma}{manyedgesinmincut}
\label{lem:manyedgesinmincut}
Let $(\flownet,d,k)$ be an instance of \minkdmc and   $(\flownet^*,d,k)$ be the equivalent instance constructed in \Cref{sec:G*}.
If $G^*$ has a \mincut with at least $(kd)^2$ edges, then $(\flownet^*,d,k)$ %
is a \emph{yes}-instance of \minkdmc.
Moreover, we can output $k$ $d$-diverse \mincuts in polynomial time.
\end{restatable}

\begin{proof}
Let $(A,B)$ be a \mincut of $G^*$ with at least $(kd)^2$ edges and let $S=E_{G^*}(A,B)$.
By Lemma~\ref{lem:ordering}, there exist an ordering $s=v_1, \dots, v_n=t$ of $V(G^*)$ and an index $1\le i\le n$ such that $A=\{v_1,\dots,v_i\}$ and $B=\{v_{i+1},\dots,v_n\}$.

Since $ G^*$ has no multiple edges,
  either the left side $\{v_1,\dots,v_i\}$ contains at least $d(k-1)+1$
  endpoints of edges in~$S$, or the right side $\{v_{i+1},\dots,v_n\}$
  does (indeed, if each side has at most $d(k-1)< kd$ endpoints, then we have $|S| < (dk)^2$).  Assume w.l.o.g.\ the left side contains at least $d(k-1)+1$ endpoints.  Let
  $L$ be any set of $d(k-1)+1$ endpoints of $S$ on the left side.

  Define indices $i_1,\dots,i_k$ so that $i_1<\dots <i_k$, $i_1$ is the minimum index of the vertices in $L$, and for each $j\in[2,k]$, $v_{i_j}$ is the endpoint in $L$ such that
  exactly $d-1$ vertices of $L$ are between $v_{i_{j-1}}$ and $v_{i_j}$ in the ordering. Since $|L|\geq d(k-1)+1$, it implies that $i_k\leq i$. For each $j\in[k]$, we set
  $A_j=\{v_1,\ldots,v_{i_j}\}$ and $B_j=\{v_{i_j+1},\ldots,v_n\}$ (see~\Cref{fig:kdsquaredyes_2}).
  By \Cref{lem:ordering}, each $(A_j,B_j)$ is a \mincut.
  For every $j,j'\in[k]$ such that $j<j'$, we have that $E_{G^*}(A_{j'},B_{j'})$ contains at least $d$  edges of $S$ that are not in $E_{G^*}(A_j,B_j)$. Thus, for $C_j=E_{G^*}(A_j,B_j)$ and $C_{j'}=E_{G^*}(A_{j'},B_{j'})$,
  $|C_{j'}\triangle C_j|\geq d$. This means that $C_1,\ldots,C_k$ are $d$-diverse.

  Therefore, if $|S|\ge(dk)^2$, the instance is a \emph{yes}-instance.
\end{proof}

We are finally ready to prove the main preprocessing lemma.

\begin{proof}[Proof of \Cref{lem:preprocess}]
We find a \maxflow in~$G$ in polynomial time.
Then by making use \Cref{obs:saturatedcritical}, we check for every edge $e$, whether it is critical by calling the flow algorithm for $G-e$. This way, we construct the sets $C$ and $R$ in $\cO(m^{2 + o(1)})$ time.
In the next step, we contract the residual edges that have positive flow through them.
We then compute in linear time the strongly connected components of $G_R'$~\cite{tarjan1972depth}, and contract them.
This way, in time $\cO(m^{2 + o(1)})$, we construct the instance $(\flownet^*,d,k)$ as described in \Cref{sec:G*}.
By \Cref{lem:mincutsaftercontraction}, $(\flownet,d,k)$ and $(\flownet^*,d,k)$ are equivalent instances of \minkdmc.
In time $\cO(m^{1 + o(1)})$, we apply \Cref{prop:findingmaxedgemincut} with input $\flownet^*$.
It outputs a \mincut $(A,B)$ of $G^*$ whose number of edges is maximized.
If $(A,B)$ has at least $(kd)^2$ edges, then, by \Cref{lem:manyedgesinmincut}, $(\flownet^*,d,k)$, and thus $(\flownet,d,k)$, is a yes-instance of \minkdmc.
Otherwise, we return $(\flownet^*,d,k)$, given that any \mincut of $G^*$ has at most $(kd)^2$ edges.
\end{proof}

\section{Algorithm for Solving Bounded Cuts}

In this section, we restate and prove \Cref{lem:mainstep}.

\begin{lemma}
\label{lem:mainstep}
There is an algorithm that solves \minkdmc in time $2^{k^2r}\cdot(dkr)^{\cO(kr)}\cdot r^{\cO(r^4)}\cdot m n \log m$, where $r$ is the maximum number of edges of a \mincut of the input graph.
\end{lemma}

\subsection{Color Coding}

We now describe the color-coding scheme that will allow us to isolate the
relevant cut structure.
Let us begin with some definitions.
Let $(\flownet,d,k)$ be a yes-instance of \minkdmc and let $\mathcal{C} = (C_1,\dots,C_k)$ be $d$-diverse \mincuts in $G$.
Let $U = \cutproduct{C}$.

\smallskip
\noindent
\textbf{Membership encoding.}
Each edge $e\in E$ has a membership vector
\[
\beta(e)\in\{0,1\}^k, \qquad
\beta(e)_i = 1 \iff e\in C_i.
\]
The set of edges of $U$ corresponding to $\alpha\in\{0,1\}^k$ is defined by
\[
  C^{\cap}_\alpha := \{e \in U : \beta(e) = \alpha\}.
\]

\smallskip
\noindent
\textbf{Active regions.}
We define the set of \emph{active regions} of $\mathcal{C}$ by
\[
\mathcal{R}_{\mathrm{act}}
:=\{\alpha\in\{0,1\}^k : C^{\cap}_\alpha\neq\emptyset\}.
\]
That is, $\alpha\in\mathcal{R}_{\mathrm{act}}$ if and only if region $C^{\cap}_\alpha$ contains at least one edge
that belongs to at least one of the cuts in $\mathcal{C}$.

\begin{observation}
\label{obv:active-regions}
Assume every \mincut $C$ in $G$ satisfies $|C|\le r$.
Then the family $\mathcal{C}= (C_1,\dots,C_k)$ has at most $kr$ active regions:
\[
|\mathcal{R}_{\mathrm{act}}|\le kr.
\]
\end{observation}

\noindent
Note that the number of choices for the set of active regions is thus at most $(2^k)^{kr}=2^{k^2r}$.

\smallskip
\noindent
\textbf{Colorings.}
We use color coding with $kr$ colors, identified with vectors $\alpha\in\mathcal{R}_{\mathrm{act}}$.
A coloring of $G$ is a function
$\chi:E\to\mathcal{R}_{\mathrm{act}}$.
For each $\alpha\in\mathcal{R}_{\mathrm{act}}$, the coloring induces the region (color class)
$
E_\alpha := \{e\in E : \chi(e)=\alpha\}$.
A coloring~$\chi$ is \emph{good with respect to~$\mathcal{C}$} if it
agrees with the membership vectors of all edges in~$U$, i.e.,
for each $e \in U$, we have that $\chi(e) = \beta(e)$, or, said another way,
for every $\alpha \in \{0,1\}^k$, we have $C^{\cap}_\alpha =E_\alpha\cap U$.

\smallskip
\noindent
\textbf{Number of edges and multiplicity of each region.}
For each active region $\alpha$, we set
\[x_\alpha = \sum_{e \in C^{\cap}_\alpha}\ell(e) \qquad \text{and\quad if $x_\alpha \ge d$, we set $x_\alpha=d^+$ instead}
\]
So we have $x_\alpha \in \{0,1, \dots, d-1, d^+\}=:[0,d^+]$.
Similarly, we set
$r_\alpha$ to be the number of edges in a region $C^{\cap}_\alpha$. Since the number of edges in each cut is bounded by $r$, $r_\alpha \in [0,r]$.
We call \emph{profile of $\mathcal{C}$} the collection
$\mathcal{P} := \{(x_\alpha, r_\alpha) : \alpha\in\mathcal{R}_{\mathrm{act}}\}$.

More generally, we call \emph{profile with respect to $\mathcal{R}_{\mathrm{act}}$} any collection
$\mathcal{P} := \{(l_\alpha, n_\alpha) : \alpha\in\mathcal{R}_{\mathrm{act}},l_\alpha\in[0,d^+],n_\alpha\in[0,r]\}$.
Note that the number of profiles with respect to $\mathcal{R}_{\mathrm{act}}$ is thus at most $((r+1)(d+1))^{kr}$.

For convenience, we set $x_\alpha= r_\alpha=0$ and $E_\alpha=\emptyset$ for each $\alpha\in\{0,1\}^k\setminus \mathcal{R}_{\mathrm{act}}$, and we call \emph{profile} any profile with respect to some set of active regions $\mathcal{R}_{\mathrm{act}}$.

\medskip
Note that by the $d$-diversity of  $\mathcal{C}$, the multiplicity summation across regions where exactly either $C_1$ or $C_2$ intersects (for any two cuts $C_i,C_j \in \mathcal{C}$) takes value at least $d$. Formally, we have the following:
\begin{restatable}{observation}{obsdiverse}
  \label{obs:diverse}
  For each $i,j\in[k]$ such that $i \neq j$, $\sum_{\alpha: \alpha_i\neq \alpha_j} x_\alpha\ge d$.
\end{restatable}

The main result from this section and
this part of the algorithm is the following lemma on deterministic guessing. We apply a deterministic version of color coding as follows to guess a coloring that is good with respect to $\mathcal{C}$ and its profile.

\begin{restatable}{lemma}{profileguessdeterministic}
  \label{lem:profileguessdeterministic}
  Let $(\flownet,d,k)$ be an instance of \minkdmc, let $\ract$ be a set of
  active regions and let $\cal P$ be a profile with respect to $\ract$.
  Then one can construct in time $(kr)^{\cO(kr)}m\log m$ a set $\mathcal{H}$ of
  colorings of $G$ of size at most $(kr)^{\cO(kr)}\log m$ such that, for every
  $d$-diverse \mincuts $\mathcal{C}= (C_1,\dots,C_k)$ in $G$ with profile
  $\cal P$, there exists a coloring $\chi\in \mathcal{H}$ that is good with
  respect to $\mathcal{C}$.
\end{restatable}

\begin{proof}
     We apply a deterministic version of color coding using standard splitter constructions; we refer to the textbook on parameterized algorithms~\cite{cygan2015parameterized} for the introduction.
In particular, by known results on $(n,s)$-perfect hash families, one can find in time $e^s \cdot s^{\cO(\log s)} m\log m$ an explicit family $\mathcal F$ of functions $\lambda:E\to[s]$ of size
$e^s \cdot s^{\cO(\log s)} m\log m$ such that for every subset $D\subseteq E$ of size at most $s$, there exists some $\lambda\in\mathcal F$ such that $\lambda$ is bijective on $D$.
In particular, for $s=kr$, this is the case for $D=\bigcup_{i\in[k]}C_i$ for any \mincuts $(C_1,\dots,C_k)$.
Let $\mathcal{G}$ represent the set of all functions from $[s]$ to $\ract$ and let $\mathcal{H} :=\mathcal{G} \circ \mathcal{F} = \{ g \circ f \mid g \in \mathcal{G},\ f \in \mathcal{F} \}$.
We claim that the statement holds with the family $\mathcal{H}$.
Let $\mathcal{C}= (C_1,\dots,C_k)$ be $d$-diverse \mincuts in $G$ with profile $\cal P$.
Since at most $kr$ edges belong to $U:=\bigcup_{i\in[k]}C_i$, there is a function $f\in\cal F$ that is bijective from $U$ to $[kr]$ and furthermore, there exists a function $g \in \mathcal{G}$ that labels these edges to the correct active region in $\ract$.
Hence $\chi=g \circ f$ is a coloring $\chi\in \mathcal{H}$ that is good with respect to $\mathcal{C}$.

The size of $\mathcal{H}$ is bounded by $(kr)^{kr} \cdot e^{kr} \cdot (kr)^{\cO(\log kr)} \log m$ and takes time $(kr)^{kr} \cdot e^{kr} \cdot (kr)^{\cO(\log kr)} \log m$ to compute. Upon simplification, we obtain $|\mathcal{H}|=(kr)^{\cO{(kr)}}\log m$, and $\mathcal{H}$ is computed in time $(kr)^{\cO{(kr)}}m\log m$.

\end{proof}

\subsection{Flow Augmentation}

The \emph{flow augmentation} technique designed by Kim et al.~\cite{kim2025flow} transforms minimal cuts into minimum cuts. The blackbox result (in its deterministic version) is as follows:

\begin{theorem}[{Flow Augmentation~\cite{kim2025flow}}]
\label{thm:flowaugdeterm}
There exists an algorithm that, given a directed graph $G$, two vertices $s,t\in V(G)$,
and an integer $\lambda$, in time $2^{\cO(\lambda^4\log \lambda)}n^{\cO(1)}$ outputs a set
$\mathcal{A}\subseteq 2^{V(G)\times V(G)}$ of size $2^{\cO(\lambda^4\log \lambda)}(\log n)^{\cO(\lambda^3)}$
such that for every minimal $s$--$t$ cut $Z \subseteq E(G)$ of size at most $\lambda$ there exists
$A\in\mathcal{A}$ such that $Z$ remains an $s$--$t$ cut in $G + A$ and, furthermore, $Z$
is a \mincut in $G + A$.
\end{theorem}

Note that, given that $\lambda^3 \log\log n\le\lambda^4$ if $\log\log n\le\lambda$ and $\lambda^3 \log\log n\le(\log\log n)^4$ otherwise, we have

\begin{align*}
|\mathcal{A}|=
2^{\cO(\lambda^4 \log \lambda)}(\log n)^{\cO(\lambda^3)}
&=
2^{\cO(\lambda^4 \log \lambda)} \cdot 2^{\cO(\lambda^3 \log\log n)}\\
&\leq
2^{\cO(\lambda^4 \log \lambda)} \cdot 2^{\cO(\lambda^4 + (\log\log n)^4)}
=
2^{\cO(\lambda^4 \log \lambda)}\, n^{o(1)}.
\end{align*}

We use   flow augmentation   to find minimum cuts with the required colors and specifications.

\begin{lemma}\label{lem:colored_cuts}
There is an algorithm that, given a directed graph $G$, an edge capacity $c$, a coloring $\chi$ of $G$, a profile $\mathcal{P} := \{(x_\alpha, r_\alpha) : \alpha\in\{0,1\}^k\}$, and an integer $r$,
find in time $k^{r+1}r^{\cO(r^4)}d^rmn$, if they exist, for each $i \in [k]$, a minimum cut $C_i'$ such that,
\begin{itemize}
    \item $C_i'\subseteq \bigcup_{\alpha:\alpha_i=1}E_\alpha$ and $\lambda_i =|C_i'|\le r$, and
    \item for each active $\alpha$, $\sum_{e \in C_i' \cap E_\alpha }\ell(e) = x_\alpha$ and $|C_i'\cap E_\alpha|=r_\alpha.$
\end{itemize}
\end{lemma}

\begin{proof}
Let $\mathbf{1}:E(G)\to \{1\}$ be the all-ones capacity function.
The deterministic version of flow augmentation (Theorem~\ref{thm:flowaugdeterm}) guarantees a set $\mathcal{A}$ such that for every $C$ that is minimal in $(G,\mathbf{1})$, there exists $A \in \mathcal{A}$ such that $C$ is an \mincut in $(G+A, \mathbf{1})$. For $(G,\mathbf{1})$, use Theorem~\ref{thm:flowaugdeterm} to obtain $\mathcal{A}$.

We take a step back to see why minimal cuts come into the picture.

\begin{claim}\label{claim:minimaltominimum}
  If $C$ is a \mincut in $(G,c)$, then $C$ is a minimal $s$--$t$ cut in $(G,\mathbf{1})$.
\end{claim}
\begin{claimproof}
    A \mincut is inclusion-wise minimal. If $C$ is not minimal, let $C' \subset C$ such that $C'$ is an $s$--$t$ minimal cut in $(G,\mathbf{1})$. This means that $C'$ is a cut in $(G,c)$ such that $c(C') < c(C)$. This contradicts the fact that $C$ is an $s$--$t$ mincut.
\end{claimproof}

Onto the procedure, for each $i\in[k]$, the iteration consists of the following steps.

\smallskip
\noindent\textbf{Step 1 (The augmentation of edges for the cut).}
 Since $C_i'$ is minimal in $(G, \mathbf{1})$ by Claim~\ref{claim:minimaltominimum}, we can find $A_i \in \mathcal{A}$ such that $C_i'$ is a \mincut in $G+A_i$. This $A_i$ can be found in $2^{\cO(\lambda_i^4 \log \lambda_i)}\, n^{o(1)}$ guesses by \Cref{thm:flowaugdeterm}. From our bounds on $\lambda_i$, this is $2^{\cO(r^4 \log r)}\, n^{o(1)}$.

\smallskip
\noindent\textbf{Step 2 (Menger paths).}
 Recall that $\lambda_i$ represents the number of edges expected in $C_i'$, which is a \mincut in $(G+A_i,\mathbf{1})$. By Lemma~\ref{lem:mengerdecompostion}, we construct $P_1, \dots, P_{\lambda_i}$, the edge-disjoint $s$--$t$ paths in $G+A_i$.

\smallskip
\noindent\textbf{Step 3 (Assigning types to paths).}
A \mincut takes one edge from each $s$-$t$ path $P_j$. Each edge in $C_i'$ is painted with a color and multiplicity $(\chi(e), \ell(e))$. This means that we are required to guess the color and the multiplicity of the edge from each path $P_j$. The multiplicity can take values from $\{1, \dots, d-1,d^+\}$ and there are at most $kr$ colors to choose from.

Naturally, guesses that do not satisfy the conditions of \Cref{lem:colored_cuts} are discarded. The relevant conditions to be met on the guesses are highlighted below. For $1 \le j \le \lambda_i$,

\begin{enumerate}
    \item\label{enum:cond1} $(\alpha_j)_i = 1$ (so that $C_i'\subseteq \bigcup_{\alpha:\alpha_i=1}E_\alpha$),
    \item\label{enum:cond2} for all active $\alpha$ such that $(\alpha)_i = 1$, if $x_\alpha = d^+$ then $\sum_{j:\alpha_j = \alpha}l_j \ge d$ and otherwise $\sum_{j:\alpha_j = \alpha}l_j = x_\alpha$ (so that $\sum_{e \in C_i' \cap E_\alpha }\ell(e) = x_\alpha$), and
    \item\label{enum:cond3} For all active $\alpha$ such that $(\alpha)_i = 1$,
    $|\{j:\alpha_j = \alpha\}| = r_\alpha$ (so that $|C_i'\cap E_\alpha|=r_\alpha$).
\end{enumerate}

\smallskip
\noindent\textbf{Step 4 (Capacity constraints to compute a cut).}
We assign capacities to the graph $G+A_i$. Define $c_i(e) = c(e)$ if $e\in P_j$ for some $j$ and if $(\chi(e), \ell(e))=(\alpha_j,l_j)$. Otherwise, set $c_i(e) = +\infty$. The graph $(G+A_i, c_i)$ is the auxiliary graph.

\smallskip
\noindent\textbf{Step 5 (Reconstruction of the cuts).}
Compute a \mincut $C_i'$ in $(G+A_i,c_i)$ and check if its capacity equals $\mu$. If not, return \emph{no}-instance.
Given that $C_i'$ is a \mincut in $(G+A_i,c_i)$, it is also a \mincut in $(G,c)$ by construction.
Additionally, since bad guesses were discarded at Step~3, we know that $C_i'$ follows the requirements of the statement.

Hence, after going through Steps 1--5 for each $i\in[k]$, if $\cutenum{C'}$ are the resulting cuts then we report a \emph{yes}-instance and output these cuts. Otherwise, we report a \emph{no}-instance.

\smallskip
\noindent
\textbf{Correctness.}
Suppose that there exist $\cutenum{C}$ following the requirements of the statement.
Let us argue that the algorithm returns a yes-instance.
Consider the graph $(G,\mathbf{1})$. By Claim~\ref{claim:minimaltominimum}, each $C_i$ is a minimal cut in $(G,\mathbf{1})$.
In Step 1, by Theorem~\ref{thm:flowaugdeterm}, there exists a set $\mathcal{A}$ such that for some $A_i \in \mathcal{A}$, the cut $C_i$ is a minimum cut in $G + A_i$. The procedure enumerates all such $A_i$, hence it considers the correct one.
Step 2 constructs $P_1, \dots, P_{\lambda_i}$ edge disjoint $s$--$t$ paths. We know that $C_i$ takes exactly one edge from each path $P_j$. Let edge $e_j$ denote the edge selected from $P_j$ by $C_i$.
In Step 3, the individual color and the multiplicity $(\chi(e_j), \ell(e_j))$ is guessed for each path $P_j$ and this guess is relevant and is not discarded since conditions \ref{enum:cond1}, \ref{enum:cond2} and \ref{enum:cond3} are satisfied.
In Step 4, for a path $P_j$, only the edges with color and multiplicity $(\chi(e_j), \ell(e_j))$ have their finite capacities $c$ assigned. The rest of the edges have infinite capacities. Under this assignment, $C_i$ remains feasible.
In Step 5, we know that $C_i$ satisfies as a minimum cut for $(G+A_i,c_i)$. Indeed, a cut with lower capacity would also be a cut in $(G,c)$, a contradiction. Let $C'_i$ be the cut obtained by running the minimum cut algorithm on $(G+A_i,c_i)$.
Repeating these steps for each $i\in[k]$, we obtain $\cutenum{C'}$ and this is output by the algorithm.

\subparagraph{Running time.}
The flow augmentation requires to be run only once and computation of $\A$ takes time $2^{(r)^4\log r}n^{\cO(1)}$. For each $i \in [k]$, the running time is as follows. We know that $|\A| \le 2^{(r)^4\log r}n^{o(1)}$. We guess $A_i \in \A$, and then compute Menger paths in time $\cO(m\sqrt{n})$ (see \Cref{lem:mengerdecompostion}). The assignment of types to the paths takes at most $(kr)d$ guesses for each path and there can be at most $r$ paths in total.
This takes $\cO((krd)^{r})$ time in total. Following this, the \mincut computation for each guess can be done in time $\cO(m^{1 + o(1)})$~\cite{chen2025maximumflow}. In total, we run this procedure over $k$ iterations and end up using time
\[
  2^{\cO(r^4 \log(r))}
  \cdot k
  \cdot n
  + 2^{\cO(r^4 \log(r))}
  \cdot k
  \cdot (krd)^{r}
  \cdot O(m\sqrt n).
\]

This simplifies to $\cO(k \cdot 2^{r^4 \log(r)} \cdot \big(krd\big)^{r} \cdot m\sqrt n)$.
Upon simplifying logs and exponentials, the running time is $\cO(k^{r+1}r^{\cO(r^4)}d^r m\sqrt n)$.

\end{proof}

\begin{lemma}\label{lem:diverse'}
Let $\chi$ be a coloring of $G$ and $\mathcal{P} := \{(x_\alpha, r_\alpha) : \alpha\in\{0,1\}^k\}$ be a profile such that,
for each $i,j\in[k]$, $\sum_{\alpha: \alpha_i\neq \alpha_j} x_\alpha\ge d$.
Then the cuts $C_1',\dots,C_k'$ output by \Cref{lem:colored_cuts} are $d$-diverse.
\end{lemma}

\begin{proof}
Let us show that $C_i'$ and $C'_j$ are $d$-diverse for each $i\ne j$. We have:
\begin{align*}
\ell(C_i'\triangle C_j') %
&\ge \sum_{\alpha:\alpha_i=1, \alpha_j=0}\sum_{e \in C_i' \cap E_\alpha} \ell(e)
 +\sum_{\alpha:\alpha_i=0, \alpha_j=1}\sum_{e \in C_j' \cap E_\alpha}\ell(e)\\
                       &= \sum_{\alpha:\alpha_i=1, \alpha_j=0}x_\alpha + \sum_{\alpha:\alpha_i=0, \alpha_j=1}x_\alpha
                       =  \sum_{\alpha:\alpha_i\neq\alpha_j} x_\alpha \ge d
\end{align*}
\end{proof}

\subsection{Correctness of the Algorithm}

We can now prove \Cref{lem:mainstep}.

\begin{proof}[Proof of \Cref{lem:mainstep}]
Our algorithm goes as follows.
For each set of active regions $\mathcal{R}_{\mathrm{act}}$ (at most $2^{k^2r}$ such sets) and for each profile $\Pcal$ with respect to $\mathcal{R}_{\mathrm{act}}$ (at most $((r+1)(d+1)^{kr}$ such profiles), we do the following. Firstly, we only retain the profiles for which \Cref{obs:diverse} holds and discard the rest. Using Lemma~\ref{lem:profileguessdeterministic}, we construct in time $(kr)^{\cO(kr)}m\log m$ a set $\mathcal{H}$ of colorings of $G$ of size at most $(kr)^{\cO(kr)}\log m$ such that, for every $d$-diverse \mincuts $\mathcal{C}= (C_1,\dots,C_k)$ in $G$ with profile $\cal P$, there exists a coloring $\chi\in \mathcal{H}$ that is good with respect to $\mathcal{C}$.

For each coloring $\chi\in\cal H$, we invoke the procedure of Lemma~\ref{lem:colored_cuts}
to test whether there exists a family of cuts $\cutenum{C'}$ such that, for each $i\in[k]$,
\begin{itemize}
    \item $C_i'\subseteq \bigcup_{\alpha:\alpha_i=1}E_\alpha$ and $\lambda_i :=|C_i'|\le r$, and
    \item $\sum_{e \in C_i' \cap E_\alpha }\ell(e) = x_\alpha$ and $|C_i'\cap E_\alpha|=r_\alpha.$
\end{itemize}
If any call succeeds, output the corresponding cuts.
If no profile yields a solution, report that no such family of cuts exists.

\subparagraph{Correctness.}
If $(\flownet,d,k)$ is a yes-instance, then there exists $k$ $d$-diverse \mincuts $(C_1,\dots,C_k)$.
At some point in the algorithm, we consider the set of active regions $\mathcal{R}_{\mathrm{act}}$ of $(C_1,\dots,C_k)$ and the profile $\mathcal{P}= \{(x_\alpha, r_\alpha) : \alpha\in\mathcal{R}_{\mathrm{act}}\}$ of $(C_1,\dots,C_k)$.
Additionally, by Lemma~\ref{lem:profileguessdeterministic}, there is a coloring $\chi$ of $G$ in $\cal H$ that is good with respect to $(C_1,\dots,C_k)$.
Hence, by Lemma~\ref{lem:colored_cuts}, we should find $k$ cuts $\cutenum{C'}$ with the correct requirement (since $\cutenum{C}$ follows those requirements).
Since \Cref{obs:diverse} holds, for each $i\in[k]$, we have $\sum_{\alpha: \alpha_i\neq \alpha_j} x_\alpha\ge d$.
This implies, by \Cref{lem:diverse'}, that $\cutenum{C'}$ are $d$-diverse.
Hence the result.

\subparagraph{Running time.} From \Cref{lem:profileguessdeterministic} and \Cref{lem:colored_cuts}, the total running time is:
\[2^{\cO(k^2r)}\cdot((r+1)(d+1))^{kr}\cdot\left((kr)^{\cO(kr)}m\log m+(kr)^{\cO{(kr)}}\log m\cdot k^{r+1}r^{O(r^4)}d^rm \sqrt n\right)\]
which, after simplification, gives
$\cO(2^{\cO(k^2r)}\cdot(dkr)^{\cO(kr)}\cdot r^{\cO(r^4)}\cdot m \sqrt n \log n)$.

\end{proof}

\section{Conclusion}
We conclude by stating some open problems.
In our paper, we proved that the \minkdmclong problem is \FPT when parameterized by~$k$ and~$d$. However, the running time of our algorithm is superexponential. It is natural to ask whether it is possible to do better. The existing complexity lower bound of de Berg, López Martínez, and Spieksma~\cite{deberg2024findingdiverse} does not exclude existence of single-exponential in~$k$ and~$d$ algorithm. Furthermore, it is interesting to know whether we can get a substantially better running time for the problem on undirected graphs, as this case may be easier. In particular, we can contract all residual edges in the preprocessing step. Then for the obtained instance, it can be shown that either it is a yes-instance of the problem or the pathwidth of the graph does not exceed~$kd$.
Another natural question is whether \minkdmclong admits a polynomial kernel for the parameterization by~$k$ and~$d$.

Further questions concern lower bounds. For~$d\in\{1,2\}$, \minkdmclong is equivalent to deciding whether the input weighted graph has~$k$ distinct minimum~$s$--$t$ cuts. Since the problem of counting \mincuts is \sharpp-complete by the result of Provan and Ball~\cite{provan1983complexitycounting}, we obtain that \minkdmclong is intractable even if~$d\in\{1,2\}$. However, this is not the case when~$k$ is polynomial in the graph size because \mincuts can be enumerated with polynomial (linear) delay~\cite{ProvanS96,TsukiyamaSOA80}. Is \minkdmclong \NP-complete in the strong sense for a given constant~$d\geq 3$, that is, when~$k$ is polynomial in~$n$?

Finally, what can be said for the dual parameterization, that is, when the diversity measure is the maximum size of the intersection of two cuts?  Is the problem \FPT in this case? Notice that the maximum family of disjoint~$s$-$t$-cuts can be found in polynomial time by the results of Wagner~\cite{wagner1990disjointcuts}.

\end{document}

%% file: Cuts.pdf_t
\begin{picture}(0,0)%
\includegraphics{Cuts.pdf}%
\end{picture}%
\setlength{\unitlength}{3947sp}%
\begingroup\makeatletter\ifx\SetFigFont\undefined%
\gdef\SetFigFont#1#2#3#4#5{%
  \reset@font\fontsize{#1}{#2pt}%
  \fontfamily{#3}\fontseries{#4}\fontshape{#5}%
  \selectfont}%
\fi\endgroup%
\begin{picture}(4759,2415)(354,-1414)
\put(4153,-1336){\makebox(0,0)[lb]{\smash{{\SetFigFont{12}{14.4}{\rmdefault}{\mddefault}{\updefault}{\color[rgb]{0,0,0}$(A,B)$}%
}}}}
\put(825,-796){\makebox(0,0)[lb]{\smash{{\SetFigFont{12}{14.4}{\rmdefault}{\mddefault}{\updefault}{\color[rgb]{0,0,0}$v_{j_1}$}%
}}}}
\put(1576,-1111){\makebox(0,0)[lb]{\smash{{\SetFigFont{12}{14.4}{\rmdefault}{\mddefault}{\updefault}{\color[rgb]{0,0,0}$d-1$}%
}}}}
\put(2185,-803){\makebox(0,0)[lb]{\smash{{\SetFigFont{12}{14.4}{\rmdefault}{\mddefault}{\updefault}{\color[rgb]{0,0,0}$v_{j_2}$}%
}}}}
\put(3518,-803){\makebox(0,0)[lb]{\smash{{\SetFigFont{12}{14.4}{\rmdefault}{\mddefault}{\updefault}{\color[rgb]{0,0,0}$v_{j_k}$}%
}}}}
\put(4324,-803){\makebox(0,0)[lb]{\smash{{\SetFigFont{12}{14.4}{\rmdefault}{\mddefault}{\updefault}{\color[rgb]{0,0,0}$v_{i+1}$}%
}}}}
\put(3978,-803){\makebox(0,0)[lb]{\smash{{\SetFigFont{12}{14.4}{\rmdefault}{\mddefault}{\updefault}{\color[rgb]{0,0,0}$v_{i}$}%
}}}}
\put(369,-784){\makebox(0,0)[lb]{\smash{{\SetFigFont{12}{14.4}{\rmdefault}{\mddefault}{\updefault}{\color[rgb]{0,0,0}$s$}%
}}}}
\put(736,-1329){\makebox(0,0)[lb]{\smash{{\SetFigFont{12}{14.4}{\rmdefault}{\mddefault}{\updefault}{\color[rgb]{0,0,0}$(A_1,B_1)$}%
}}}}
\put(2026,-1330){\makebox(0,0)[lb]{\smash{{\SetFigFont{12}{14.4}{\rmdefault}{\mddefault}{\updefault}{\color[rgb]{0,0,0}$(A_2,B_2)$}%
}}}}
\put(3375,-1335){\makebox(0,0)[lb]{\smash{{\SetFigFont{12}{14.4}{\rmdefault}{\mddefault}{\updefault}{\color[rgb]{0,0,0}$(A_k,B_k)$}%
}}}}
\end{picture}%